\documentclass[11pt]{article}

\usepackage{amsthm,amsmath,amssymb}
\usepackage{fullpage}
\usepackage{authblk}

\DeclareMathOperator*{\argmin}{arg\,min}

\newtheorem{lemma}{Lemma}
\newtheorem{theorem}{Theorem}

\begin{document}

\title{Distance problems within Helly graphs and $k$-Helly graphs}
\author[1,2]{Guillaume Ducoffe}
\affil[1]{\small National Institute for Research and Development in Informatics, Romania}
\affil[2]{\small University of Bucharest, Romania}
\date{}

\maketitle

\begin{abstract}
The {\em ball hypergraph} of a graph $G$ is the family of balls of all possible centers and radii in $G$. It has {\em Helly number} at most $k$ if every subfamily of $k$-wise intersecting balls has a nonempty common intersection. A graph is {\em $k$-Helly} (or {\em Helly}, if $k=2$) if its ball hypergraph has Helly number at most $k$. We prove that a central vertex and {\em all} the medians in an $n$-vertex $m$-edge Helly graph can be computed w.h.p. in $\tilde{\cal O}(m\sqrt{n})$ time. Both results extend to a broader setting where we define a non-negative cost function over the vertex-set. For any fixed $k$, we also present an $\tilde{\cal O}(m\sqrt{kn})$-time randomized algorithm for radius computation within $k$-Helly graphs. If we relax the definition of Helly number (for what is sometimes called an ``almost Helly-type'' property in the literature), then our approach leads to an approximation algorithm for computing the radius with an additive one-sided error of at most some constant.
\end{abstract}

\section{Introduction}\label{sec:intro}

For any undefined graph terminology, see~\cite{BoM08,Die12,Gol04}.
We study the complexity of computing several distance parameters within some increasing hierarchy of graph classes.
Specifically, let $G=(V,E)$ be an undirected graph. The distance between two vertices $u$ and $v$, denoted by $dist(u,v)$, is equal to the minimum number of edges on a $uv$-path in $G$. The eccentricity of a vertex $v$ is defined as $e(v) = \max_{u \in V} dist(u,v)$. The total distance of a vertex $v$ is defined as $TD(v) = \sum_{u \in V} dist(u,v)$. 
We focus on two fundamental {\em facility location} problems:
\begin{itemize}
\item The {\sc Center} problem asks for a vertex $v$ such that $e(v)$ is minimized ({\it a.k.a.}, central vertex).
\item The {\sc Median} problem asks for a vertex $v$ such that $TD(v)$ is minimized ({\it a.k.a.}, median).
\end{itemize}
Facility location problems are about finding optimal locations, w.r.t. some distance and cost criteria, for opening new facilities such as schools or hospitals. Due to their practical relevance, such problems have been widely studied in mathematical programming, both on metric spaces and graphs. However, our motivation for studying these two above problems also comes from their importance in the fine-grained complexity study of polynomial-time solvable problems. Indeed, we observe that both the {\sc Center} and {\sc Median} problems can be solved in ${\cal O}(n^2)$ time, if we are given the distance-matrix of the graph. Computing the distance matrix, or equivalently solving the All-Pairs Shortest-Paths problem ({\sc APSP}), can be done in ${\cal O}(nm)$ time, or in $\tilde{\cal O}(n^{\omega})$ time~\cite{Sei92}, where $\omega < 2.273$ denotes the square matrix multiplication exponent. If we further assume the input graph to be sparse (or of arbitrary density, but assuming $\omega = 2$), then this is in ${\cal O}(n^{2+o(1)})$ time, that is optimal. By our previous observation, these above time complexity upper bounds also apply to the {\sc Center} and {\sc Median} problems. However, for a long time it was open whether these upper bounds are asymptotically tight. In 2016, Abboud et al. proved that it was indeed the case if we assume the so-called ``Hitting Set Conjecture'', that implies the Orthogonal Vector Conjecture~\cite{AVW16}. Specifically, their conjecture implies that for any $\varepsilon > 0$, neither the {\sc Center} nor the {\sc Median} problems can be solved in ${\cal O}(n^{2-\varepsilon})$ time on $n$-vertex $\tilde{\cal O}(n)$-edge graphs. In particular, this rules out the existence of an $\tilde{\cal O}(n^am^b)$-time algorithm, for every positive $a,b$ such that $a+b < 2$. We stress that this above conditional lower bound also holds for {\em radius} computation ({\it i.e.}, the seemingly weaker problem of computing the minimum eccentricity, but not necessarily a corresponding central vertex).

\medskip
\noindent
We follow a long line of work in Algorithmic Graph Theory, that consists in improving the complexity of distance problems on some special graph classes~\cite{AVW16,BCCV19+,BCD98,BHM18,Cab18,CaK09,CDV02,ChD94,ChD03,CDP18,CDHP01,Dam16,Dra89,DrN00,DNB97,Duc19,DHV19,DHV20,Epp00,FaP80,GKHM+18,Gol71,Ola90}. Specifically, the ball of center $v$ and radius $r$ is the set $N^r[v] = \{ u \in V \mid dist(u,v) \leq r \}$. The {\em ball hypergraph} of $G$, denoted by ${\cal B}(G)$, is the family of the balls of all possible centers and radii in $G$. That is, ${\cal B}(G) = \{ N^r[v] \mid v \in V, \ r \geq 0 \}$. Several interesting graph classes can be defined via the properties of their ball hypergraphs. For instance, the dually chordal graphs are exactly the graphs of which the ball hypergraph is a dual hypertree~\cite{BDCV98}. There exists a linear-time algorithm for computing the radius of a dually chordal graph~\cite{BCD98}. More recently, we proved in~\cite{DHV20}, with co-authors, that there exists a parameterized subquadratic-time algorithm for computing the radius, within all the graph classes whose ball hypergraphs have constantly upper-bounded VC-dimension. As a particular case, we got a truly subquadratic-time algorithm for radius computation within all the proper minor-closed classes of graphs. This work makes a new step toward finding a more general property of ball hypergraphs that we could still exploit for a fast radius computation. 

Specifically, a graph is {\em $k$-Helly} if every sub-family of $k$-wise intersecting balls has a non-empty common intersection. The $2$-Helly graphs have already received some attention in the literature, under the different names of {\em Helly graphs} (that we use in this paper), disk-Helly graphs~\cite{BaP91} or absolute retracts~\cite{BaP89}. Notably, the interval graphs, the strongly chordal graphs and the dually chordal graphs are all subclasses of the Helly graphs~\cite{BDCV98}. In fact, a celebrated result in Metric Graph Theory is that {\em every} graph is an isometric (distance-preserving) subgraph of some Helly graph~\cite{Dre84,Isb64}. Therefore, unlike most of the special graph classes studied in the literature, the class of Helly graphs {\em cannot} be defined via a set of forbidden patterns such as (induced) subgraphs, minors or bounded expansion. Note that in general, a bounded VC-dimension does {\em not} imply a bounded Helly number, but it does imply a slightly weaker {\em fractional Helly property}~\cite{Mat04}. We refer to the survey~\cite{BaC08}, that has devoted a full section to the properties of Helly graphs. See also~\cite{BaC02,BaP89,BaP89b,BaP91,CLP00,Dra89,Dra93,DrB96,DuD19+,Pol01,Pol03} for other properties, and, {\it e.g.},~\cite{BaP89,Dra89,LiS07} for improved polynomial-time recognition algorithms for the Helly graphs. For $k \geq 3$, we stress that we can decide in $n^{{\cal O}(k)}$ time whether an $n$-vertex graph is $k$-Helly, that is polynomial for every fixed $k$; indeed, a graph is $k$-Helly if and only if for every $(k+1)$-subset $S$ of vertices, the balls intersecting $S$ in at least $k$ vertices have a non-empty common intersection~\cite{Ber73}. The running-times of these recognition algorithms are super-quadratic, but we do {\em not} need to execute them before we can apply our own algorithms in this paper. 

\smallskip
\noindent
Recently, with Dragan, we initiated the fine-grained complexity study of polynomial-time solvable distance problems within Helly graphs~\cite{DuD19+}. Specifically, we proved that the radius of an $n$-vertex $m$-edge Helly graph can be computed w.h.p. in $\tilde{\cal O}(m\sqrt{n})$ time. Our result followed from a known reduction to the computation of a diametral pair, {\it i.e.}, a pair of vertices whose distance in the graph is maximum. Indeed, for a Helly graph of diameter $D$, the radius is exactly $\lceil D/2 \rceil$~\cite{Dra89}. Unfortunately, we failed in also computing a central vertex with this approach.

\subsection*{Our Contributions}

\paragraph{Helly graphs.} We further refine the framework in~\cite{DuD19+}, in order to solve the {\sc Center} and {\sc Median} problems in truly subquadratic-time within the Helly graphs.
\begin{itemize}
\item Our first result in this paper is an algorithm for computing a central vertex in a Helly graph, that runs in $\tilde{\cal O}(m\sqrt{n})$ time w.h.p. (Theorem~\ref{thm:helly-center}). For that, we use the nice property that the eccentricity function of a Helly graph is {\em unimodal}: every local minimum is a global minimum~\cite{Dra89}. The latter suggests the following natural {\em local-search} strategy in order to compute a central vertex: start from some vertex with a hopefully small eccentricity, then repeatedly find a neighbour with a better eccentricity until a local minimum is found. Obvious drawbacks to this approach are: the time needed for each step of the local search, and the total number of steps. In order to overcome the first issue, we use and further refine some ``gated'' properties of the sets of small diameter in a Helly graph. We handle the second issue with a classical random sampling technique, that allows us to find, amongst $\tilde{\cal O}(\sqrt{n})$ candidates, a start vertex which ensures a sufficiently small (sublinear) total number of steps. 
\item Our second main result is a similar algorithm for computing the median in a Helly graph, that also runs in $\tilde{\cal O}(m\sqrt{n})$ time w.h.p. (Theorem~\ref{thm:helly-median}). Indeed, since the total distance function of a Helly graph is also unimodal~\cite{BaC02}, we can use the same local-search strategy as before in order to compute a median. For a Helly graph, the median is a clique~\cite{BaC02}. By using this nice structural property, we can compute all the medians in one more step of the local search.
\end{itemize}
Our approach still works for a more general version of the {\sc Center} and {\sc Median} problems: where we assign costs to the vertices. We stress that in~\cite{DuD19+}, we used different properties of unimodal functions in order to compute a diametral pair. In an upcoming paper~\cite{DDG19+}, based on the techniques developed here and in~\cite{DuD19+}, we prove that we can also compute the center of a Helly graph in truly subquadratic time ({\it i.e.}, the set of all its central vertices). Note that for a Helly graph, computing the center is equivalent to computing all the eccentricities~\cite{Dra89}.

\paragraph{$k$-Helly graphs.} Our third main result in the paper is a randomized algorithm for computing the radius of a $k$-Helly graph, which runs w.h.p. in $\tilde{\cal O}(m\sqrt{kn})$ time (Theorem~\ref{thm:k-helly}). In particular, our result implies that the hard instances for radius computation must have a quasi linear Helly number. Furthermore, if we relax the Helly property, as described next, then we obtain an approximation algorithm for computing the radius. Specifically, let us call a graph $(k,\alpha)$-{\em Helly} if every family of $k$-wise intersecting balls has a nonempty common intersection when we increase the radius of each ball by $\alpha$. For the special case $k=2$, this relaxed property has been studied under the names of Helly-gap~\cite{DrG020} or coarse Helly property~\cite{CCGH+20}. Our algorithm outputs the radius with some additive one-sided error in ${\cal O}(\alpha)$. 

Closest to our work on the $k$-Helly graphs is a very general framework by Amenta~\cite{Ame94}. Namely, she proved that on set families of constant Helly number, deciding whether these sets have a nonempty common intersection could be rephrased, under some mild conditions, as a so-called Generalized Linear Program (GLP) with a bounded combinatorial dimension. The GLP's are a broad generalization of linear programming, for which the prototypical problem is the computation of a central point in Euclidean spaces of bounded dimension. If a GLP has a bounded combinatorial dimension, then we can solve it in randomized linear time, by using Clarkson's algorithm~\cite{Cla95}. However, the latter relies on two primitives, namely basis computation and violation testing. We do not know how to implement these two basic operations in constant time, nor even in sublinear time, for the special case of graphs. Therefore, we had to differ from Amenta's framework in order to prove our results, although we stayed closed in spirit to the latter. Our algorithm requires $k$ to be given, and it fails in also computing a central vertex. Roughly, this is because we fail in satisfying another technical condition from Amenta's framework, namely the so-called ``Unique Minimum Condition''. We left open whether, for any fixed integer $k \geq 3$, a central vertex can be computed in truly subquadratic time whithin the $k$-Helly graphs.

\section{Facility Location problems on Helly graphs}\label{sec:helly}
In what follows, let $c : V \to \mathbb{R}^+$ be a cost assignment over the vertex-set of some graph $G=(V,E)$.
We define the $c$-eccentricity and the total $c$-distance of a vertex $v$ as $e_c(v):= \max_{u \in V} c(u) \cdot dist(u,v)$ and $TD_c(v) := \sum_{u \in V} c(u) \cdot dist(u,v)$, respectively. The main results of this section are truly subquadratic algorithms, for the {\sc Center} and {\sc Median} problems, respectively, on Helly graphs (Theorems~\ref{thm:helly-center} and~\ref{thm:helly-median}). More generally, for any function $c$ that assigns nonnegative costs on the vertices, we can compute a vertex minimizing its $c$-eccentricity, resp. {\em all} vertices minimizing their total $c$-distance. For that, we first need to introduce several intermediate results on Helly graphs and unimodal functions.

\paragraph{Unimodal functions.}
Let $G=(V,E)$ be a graph. In what follows, we denote by $N(v)$ the neighbours of a vertex $v$.
Furthermore, let $f : V \to \mathbb{N}$. A local minimum for $f$ is a vertex $v$ such that, for every $u \in N(v)$, we have $f(u) \geq f(v)$. If furthermore $f(v) = \min_{w \in V} f(w)$, then $v$ is a global minimum for $f$. We denote by $\argmin{(f)}$ the set of all the global minima. Finally, we recall that $f$ is called unimodal if every local minimum is also a global minimum.

\begin{lemma}[~\cite{Dra89}]\label{lem:unimod-ecc}
If $G = (V,E)$ is a Helly graph, then for any non-negative cost function $c$, the function $e_c : V \to \mathbb{R}^+$ that maps every vertex $v$ to its $c$-eccentricity is unimodal. 
\end{lemma}

\begin{lemma}[~\cite{BaC02}]\label{lem:unimod-td}
If $G = (V,E)$ is a Helly graph, then for any non-negative cost function $c$, the function $TD_c : V \to \mathbb{R}^+$ that maps every vertex $v$ to its total $c$-distance is unimodal. Moreover, the median of $G$ is a clique. 
\end{lemma}

\noindent
By $U(p)$ we mean a vertex-subset in which every vertex was added independently at random with probability $p$. We say that a path $P = [v_0,v_1,\ldots,v_{\ell}]$ is $f$-monotone if we have $f(v_i) > f(v_{i+1})$ for every $0 \leq i < \ell$. Our starting point for this section is the following observation: 

\begin{lemma}\label{lem:unimodal}
Let $G=(V,E)$ be a graph and $f : V \to \mathbb{N}$ be unimodal.
For any random subset $U(p)$, if $u \in U(p)$ minimizes $f(u)$, then w.h.p. every $f$-monotone path between $u$ and $\argmin{(f)}$ has length at most in $\tilde{\cal O}(p^{-1})$. 
\end{lemma}

\begin{proof}
Fix a longest $f$-monotone path $P$ between $u$ and $\argmin{(f)}$, and assume that it has a length $> 1 + 2p^{-1}\log{n}$. By minimality of $f(u)$, we have $\left(V(P) \setminus \{u\}\right) \cap U(p) = \emptyset$. However, this can only happen with probability $(1-p)^{|V(P)|-1} \leq (1-p)^{\frac{2\log{n}}p} \leq n^{-2}$.
\end{proof}

\paragraph{Gates and Pseudo-gates.}
Given a graph $G$ and a vertex $v$, we recall that the ball of center $v$ and radius $r$ is denoted by $N^r[v]$.
For any vertex-subset $S$ we define $dist(v,S) = \min_{s \in S} dist(v,s)$, and we set $Pr(v,S) = \{ s \in S \mid dist(v,s) = dist(v,S)\}$.
The (weak) diameter of a set $S$ is equal to $diam(S) := \max_{x,y \in S} dist(x,y)$. 

\begin{lemma}[~\cite{DuD19+}]\label{lem:gated}
Let $G$ be a Helly graph and $S$ be a subset of weak diameter at most two. Then, for any $v \notin S$ there exists a vertex $g_S(v) \in N^{dist(v,S)-1}[v] \cap \bigcap \{ N(x) \mid x \in Pr(v,S) \}$. Moreover, we call $g_S(v)$ a {\em gate} of $v$.
\end{lemma}

In~\cite{DuD19+}, we introduce a complementary notion of ``pseudo-gates'' for the subclass of the $C_4$-free Helly graphs. The latter notion was key to our linear-time algorithm for computing the diameter of the graphs in this subclass.
We now generalize~\cite[Lemma 19]{DuD19+} to all the Helly graphs.

\begin{lemma}\label{lem:pseudo-gate}
Let $G$ be a Helly graph and let $S$ be such that $diam(S) \leq 2$. If $v \notin S$ and $dist(v,S) = r -1$, then there exists a vertex $pg_S(v) \in N^{r-1}[v]$ such that $S \cap N^{r}[v] \subseteq N[pg_S(v)]$. Moreover, $pg_S(v)$ is in the closed neighbourhood of some gate of $v$. We call $pg_S(v)$ a {\em pseudo-gate} of $v$.
\end{lemma}

\begin{proof}
Since the balls $N^{r-1}[v]$ and $N[x], \forall x \in S \cap N^r[v]$ pairwise intersect, by the Helly property there exists a vertex $pg_S(v) \in N^{r-1}[v] \cap \bigcap \{ N[x] \mid x \in S \cap N^r[v] \}$. Now, consider the balls $N^{r-2}[v], N[pg_S(v)]$ and $N[y], \forall y \in Pr(v,S)$. Again, these balls pairwise intersect, and so, by the Helly property, there exists a vertex $v^* \in N^{r-2}[v] \cap \{ N[z] \mid z \in Pr(v,S) \cup \{pg_S(v)\} \}$. By construction, $v^* \in N[pg_S(v)]$ is a gate of vertex $v$. 
\end{proof}

Let $\delta$ be the Kronecker symbol (equal to $1$ if the predicate in the subscript is true, and to $0$ otherwise).
By combining Lemmata~\ref{lem:gated} and~\ref{lem:pseudo-gate}, we are now ready to prove our main technical lemma.

\begin{lemma}\label{lem:local-step}
Let $G=(V,E)$ be a Helly graph, equipped with a non-negative cost function $c$.
For any $v \in V$ and $A \subseteq V$, we can compute the following three values for every neighbour $u \in N(v)$, in total ${\cal O}(m)$ time:
\begin{itemize}
\item $q_+(u,A) = \sum_{w \in A} c(w) \cdot \delta_{dist(u,w) \ > \ dist(v,w)}$;
\item $q_=(u,A) = \sum_{w \in A} c(w) \cdot \delta_{dist(u,w) \ = \ dist(v,w)}$;
\item $q_-(u,A) = \sum_{w \in A} c(w) \cdot \delta_{dist(u,w) \ < \ dist(v,w)}$.
\end{itemize}
\end{lemma}

\begin{proof}
Let $S = N[v]$. By Lemma~\ref{lem:gated}, every vertex $w \notin S$ has a gate. Furthermore, as first observed in~\cite{ChD94} (see also Remark $1$ in~\cite{DuD19+}), we can compute a fixed gate $g(w)$ for every vertex $w \notin S$, in total ${\cal O}(m)$ time. For every vertex $z \in N(S)$, let $\alpha(z) = \sum_{w \in A \setminus S}  c(w) \cdot \delta_{g(w) = z}$. Let $u \in S$ be an arbitrary neighbour of vertex $v$. Since every path between $w \notin S$ and $v$ must intersect $S \setminus \{v\}$, we observe that we have $dist(u,w) < dist(v,w)$ if and only if $g(w) \in N(u)$. Clearly, $u$ is the only vertex $w\in S$ such that $dist(u,w) < dist(v,w)$. Therefore, for every neighbour $u \in N(v)$, we have: 
$$q_-(u,A) = c(u) \cdot \delta_{u \in A} + \sum\limits_{z \in N(S) \cap N(u)} \alpha(z).$$ 
This computation takes time ${\cal O}(|N(u)|)$.

\smallskip
\noindent
In the same way, by Lemma~\ref{lem:pseudo-gate}, every vertex $w \notin S$ has a pseudo-gate.
We can also compute a fixed pseudo-gate $pg(w)$ for every vertex $w \notin S$, in total linear time ({\it i.e.}, see Remark $3$ in~\cite{DuD19+}). Let $w \notin S$ be fixed and let $u \in S$ be a neighbour of $v$. If $g(w) \in N(u)$ then, by definition of a pseudo-gate we also have $pg(w) \in N[u]$. Otherwise, $pg(w) \in N[u]$ if and only if $dist(u,w) = dist(v,w) = 1 + dist(w,S)$. As a result, for every vertex $z \in N[S]$, let $\beta(z) = \sum_{w \in A \setminus S}  c(w) \cdot \delta_{pg(w) = z}$. We obtain that for every neighbour $u \in N(v)$: 
$$q_-(u,A) + q_{=}(u,A) = \sum\limits_{u' \in N(v) \cap N[u] \cap A} c(u')  +  \sum\limits_{z \in N[S] \cap N[u]} \beta(z).$$ 
This computation also takes time ${\cal O}(|N(u)|)$.

\smallskip
\noindent
Finally, $q_+(u,A) =  \sum_{w \in A} c(w) - q_-(u,A) - q_=(u,A)$.
\end{proof}

\paragraph{Main Results.}

\begin{theorem}\label{thm:helly-center}
If $G$ is a Helly graph then, for any non-negative cost function $c$, w.h.p., we can compute a central vertex in $\tilde{\cal O}(m\sqrt{n})$ time.
\end{theorem}

\begin{proof}
Set $p = n^{-\frac 1 2}$, and let $U(p)$ be a corresponding random subset.
By Chernoff bounds the subset $U(p)$ has cardinality $\tilde{\cal O}(\sqrt{n})$ w.h.p., and thus we assume from now on that it is indeed the case. We compute the $c$-eccentricity for every vertex of $U(p)$, that takes ${\cal O}(m|U(p)|) = \tilde{\cal O}(m\sqrt{n})$ time. Then, let $u \in U(p)$ be of minimum $c$-eccentricity. At each step of the algorithm, we search for a neighbour $v$ of the current vertex $u$ such that $e_c(v) < e_c(u)$. If no such neighbour exists then, $u$ is a local minimum for the eccentricity function, and so, by Lemma~\ref{lem:unimod-ecc}, this vertex $u$ is central. Otherwise, we set $u := v$ and then we continue the algorithm for at least one more step. Since all the vertices $u$ considered during the algorithm induce an eccentricity-monotone path, by Lemma~\ref{lem:unimodal} the total number of steps is upper bounded w.h.p. by an $\tilde{\cal O}(\sqrt{n})$.

We can implement each step of this above local-search algorithm in ${\cal O}(m)$ time, as follows.
First we need to observe that for every vertex $u$ and any neighbour $v \in N(u)$, for every vertex $w$ we have $|dist(v,w) - dist(u,w)| \leq 1$. Then, let $A = \{ w \in V \mid c(w) \cdot dist(u,w) = e_c(u) \}$ and let $B = \{ w \in V \mid c(w) \cdot (dist(u,w) + 1) \geq e_c(u) \}$. We have $e_c(v) < e_c(u)$ if and only if $q_+(v,A) = q_=(v,A) = 0$, and in the same way $q_+(v,B) = 0$. Therefore, by applying Lemma~\ref{lem:local-step} twice, we can decide in ${\cal O}(m)$ time whether there exists a neighbour $v \in N(u)$ such that $e_c(v) < e_c(u)$, and if so, compute such a neighbour within the same amount of time.
\end{proof}

\begin{theorem}\label{thm:helly-median}
If $G$ is a Helly graph then, for any non-negative cost function $c$, w.h.p., we can compute the median in $\tilde{\cal O}(m\sqrt{n})$ time.
\end{theorem}

\begin{proof}
First, we compute a median. Since by Lemma~\ref{lem:unimod-td}, the total $c$-distance function of a Helly graph is unimodal, we can adapt the local-search algorithm of the previous Theorem~\ref{thm:helly-median}. However, we need to explain how we can implement each step of this algorithm in ${\cal O}(m)$ time. For that, let $u \in V$ be fixed and let $v \in N(u)$ be an arbitrary neighbour. Again, recall that for every vertex $w$ we have $|dist(v,w) - dist(u,w)| \leq 1$. Then, we have $TD_c(u) - TD_c(v) = q_-(v,V) - q_+(v,V)$. Therefore, by applying Lemma~\ref{lem:local-step} once, we can decide in ${\cal O}(m)$ time whether there exists a neighbour $v \in N(u)$ such that $TD_c(v) < TD_c(u)$, and if so, compute such a neighbour within the same amount of time.

Finally, let $u \in V$ be a median. Again by using Lemma~\ref{lem:local-step}, in ${\cal O}(m)$ time we can compute all the neighbours $v \in N(u)$ such that $TD_c(u) = TD_c(v)$. Since the median of a Helly graph induces a complete subgraph (Lemma~\ref{lem:unimod-td}), we computed doing so all the medians.
\end{proof}

\section{Radius computation within $k$-Helly graphs}\label{sec:k-helly}

Recall that in~\cite{DuD19+}, we proved that the radius of a Helly graph can be computed w.h.p. in $\tilde{\cal O}(m\sqrt{n})$ time. In this last section, we generalize this result to the $k$-Helly graphs. The radius of a graph $G$ is denoted in what follows by $rad(G) =^{def} \min_v e(v)$.

\begin{theorem}\label{thm:k-helly}
If $G$ is a $k$-Helly graph then, w.h.p., we can compute $rad(G)$ in $\tilde{\cal O}(m\sqrt{kn})$ time.
\end{theorem}

This above Theorem~\ref{thm:k-helly} follows from a more general result, that we state next.
Recall that we call a graph $(k,\alpha)$-Helly if, for every family of $k$-wise intersecting balls $N^{r_1}[v_1], N^{r_2}[v_2],\ldots,N^{r_s}[v_s]$, for some arbitrary $s \geq 2$, there exists a vertex $x$ such that $\forall 1 \leq i \leq s, \ dist(x,v_i) \leq r_i + \alpha$. Note that in particular, the $(k,0)$-Helly graphs are exactly the $k$-Helly graphs. It is also known that the chordal graphs are $(2,{\cal O}(1))$-Helly, and more generally the $k$-hyperbolic graphs are $(2,{\cal O}(k))$-Helly~\cite{ChE07}. 

\begin{theorem}\label{thm:almost-k-helly}
If $G$ is a $(k,\alpha)$-Helly graph then, w.h.p., we can compute an additive $+\alpha$-approximation of $rad(G)$ in $\tilde{\cal O}(m\sqrt{kn})$ time.
\end{theorem}

We observe that Theorem~\ref{thm:k-helly} is an easy corollary of Theorem~\ref{thm:almost-k-helly}.
The remaining of this section is devoted to the proof of Theorem~\ref{thm:almost-k-helly}.
For that, we need the following lemma, that is based on the same random sampling technique as for Lemma~\ref{lem:unimodal}.

\begin{lemma}[~\cite{DuD19+}]\label{lem:set-cover}
Let $G=(V,E)$ be a graph, let $r$ be a positive integer and let $\varepsilon \in (0;1)$.
There is an algorithm that w.h.p. computes a set $D\langle G; r; \varepsilon \rangle$ in $\tilde{\cal O}(m/\varepsilon)$ time with the following two properties:
\begin{itemize}
\item if $e(v) \leq r$ then $v \in D\langle G; r; \varepsilon \rangle$;
\item conversely, if $v \in D\langle G; r; \varepsilon \rangle$ then  $|N^r[v]| \geq (1-\varepsilon) \cdot n$.
\end{itemize}
\end{lemma}

Equipped with Lemma~\ref{lem:set-cover}, we are now ready to prove the following decision version of Theorem~\ref{thm:almost-k-helly}:

\begin{lemma}\label{lem:decide-rad}
Let $G=(V,E)$ be a $(k,\alpha)$-Helly graph, and let $r$ be a positive integer.
There is an algorithm that, w.h.p., runs in $\tilde{\cal O}(m\sqrt{kn})$ time, and satisfies the following two properties:
\begin{itemize}
\item If the algorithm rejects, then $rad(G) > r$;
\item If the algorithm accepts, then $rad(G) \leq r + \alpha$.
\end{itemize}
\end{lemma}

Note that by a classical dichotomic argument, Lemma~\ref{lem:decide-rad} is equivalent to Theorem~\ref{thm:almost-k-helly}.

\begin{proof}
We first describe the algorithm, before proving its correctness and then, optimizing its running time.
Let $\varepsilon > 0$ to be fixed later in the proof. We proceed as follows:
\begin{itemize}
\item We compute a set $C_0 = D\langle G; r; \varepsilon \rangle$ such as in Lemma~\ref{lem:set-cover}. Such a set $C_0$ can be computed w.h.p. in ${\cal O}(m\log{n}/\varepsilon)$ time. Furthermore, if $C_0 = \emptyset$, then $rad(G) > r$, and we are done.
\item Then, for $i = 1\ldots,k$, we select an arbitrary vertex $v_{i-1} \in C_{i-1}$ and we compute $N^{r}[v_{i-1}]$. It takes linear time. If $e(v_{i-1}) \leq r$, then $rad(G) \leq r$, and we are done. Otherwise, let $S_i = V \setminus N^r[v_{i-1}]$. We compute $C_i := C_{i-1} \cap \bigcap \{ N^r[s_i] \mid s_i \in S_i \}$. Recall that $|S_i| \leq \varepsilon n$, and so, this takes ${\cal O}(mn\varepsilon)$ time. Furthermore, if $C_i = \emptyset$, then $rad(G) > r$, and we are done.
\item Finally, if $C_k \neq \emptyset$, then we accept.
\end{itemize}
Note that if we end the algorithm during the for loop, then its output is always correct. Therefore, we only need to focus on the last step. Specifically, we claim that if $C_k \neq \emptyset$, then $rad(G) \leq r + \alpha$. Indeed, observe that $rad(G) \leq r$ if and only if we have $\bigcap \{ N^r[v] \mid v \in V \} \neq \emptyset$. Conversely, since we assume $G$ to be $(k,\alpha)$-Helly, if the balls of radius $r$ $k$-wise intersect, then $rad(G) \leq r+\alpha$. Suppose by contradiction that there exists a $k$-subset $A$ such that $\bigcap \{ N^r[a] \mid a \in A\} = \emptyset$. Clearly, we cannot have $A \subseteq N^r[v_{i-1}]$, for any $1 \leq i \leq k$. Furthermore by construction the sets $S_1,S_2,\ldots,S_k$ are pairwise disjoint. Therefore, $\forall 1 \leq i \leq k, \ |A \cap S_i| = 1$. However, by construction we also have $C_k = C_0 \cap \bigcap \{ N^k[s] \mid s \in \bigcup_{i=1}^k S_i \}$. Therefore, $\emptyset \neq C_k \subseteq \bigcap\{ N^r[a] \mid a \in A\} = \emptyset$, a contradiction.

\smallskip
\noindent
Overall, the total running time is in ${\cal O}(m\log{n}/\varepsilon + kmn\varepsilon)$. This is optimized when $\varepsilon = \Theta(\sqrt{\log{n}/(kn)})$, and then the running time is in ${\cal O}(m\sqrt{kn\log{n}})$. 
\end{proof}

\bibliographystyle{abbrv}
\bibliography{biblio}

\end{document}